\documentclass[12pt]{article}

\usepackage{color}
\usepackage{stmaryrd}
\usepackage{bbm}
\usepackage{mathrsfs}
\usepackage{amsfonts}
\usepackage{amssymb}
\usepackage{amsmath}
\usepackage{amsthm}
\usepackage{cases}
\usepackage{indentfirst}
\usepackage{lscape}
\usepackage{graphicx}
\usepackage[sort&compress,numbers]{natbib}

\usepackage{epsfig}

\title{Lump solutions to nonlinear partial differential equations via Hirota bilinear forms}

\author{
Wen-Xiu Ma\thanks{Email: mawx@cas.usf.edu} \ 
and Yuan Zhou\thanks{Email: zhouy@mail.usf.edu}
\\
{\small Department of Mathematics and Statistics,
University of South Florida, }
\\
{\small Tampa, FL 33620-5700, USA}
}

\date{}

\begin{document}

\maketitle

\numberwithin{equation}{section}

\newtheorem{theorem}{Theorem}[section]
\newtheorem{lemma}[theorem]{Lemma}
\newtheorem{corollary}[theorem]{Corollary}
\newtheorem{proposition}[theorem]{Proposition}
\newtheorem{example}[theorem]{Example}
\newtheorem{definition}[theorem]{Definition}
\newtheorem{xca}[theorem]{Exercise}
\newtheorem{remark}[theorem]{Remark}

\newcommand{\abs}[1]{\lvert#1\rvert}

\newcommand{\blankbox}[2]{%
  \parbox{\columnwidth}{\centering
    \setlength{\fboxsep}{0pt}%
    \fbox{\raisebox{0pt}[#2]{\hspace{#1}}}%
  }%
}

\newcommand{\di}{\displaystyle}
\newcommand{\dist}{\operatorname{dist}}
\newcommand{\ind}{\indent}
\newcommand{\ndd}{\noindent}
\newcommand{\beq}{\begin{equation}}
\newcommand{\eeq}{\end{equation}}
\newcommand{\bl}{\label}

\def\be{\begin{equation}}
\def\ee{\end{equation}}
\def\bea{\begin{eqnarray}}
\def\eea{\end{eqnarray}}
\def\bean{\begin{eqnarray*}}
\def\eean{\end{eqnarray*}}

\def\a{\alpha}
\def\b{\beta}
\def\d{\delta}
\def\D{\Delta}
\def\e{\varepsilon}
\def\f{\varphi}
\def\F{\Phi}
\def\g{\gamma}

\newcommand{\diag}{\mathop{\rm diag}}

\begin{abstract}

Lump solutions are analytical rational function solutions localized in all directions in space.
We analyze a class of lump solutions, generated from quadratic functions, to nonlinear partial differential equations.
The basis of success is the Hirota bilinear formulation and the primary object is the class of positive multivariate quadratic functions.
 A complete determination of quadratic functions positive in space and time is given, and 
positive quadratic functions are characterized as sums of squares of linear functions. 
Necessary and sufficient conditions for positive quadratic functions to solve Hirota bilinear equations are presented, and
such polynomial solutions yield lump solutions to nonlinear partial differential equations under the dependent variable transformations 
$u=2(\ln f)_x$ and $u=2(\ln f)_{xx}$, where $x$ is one spatial variable. Applications are made for a few generalized KP and BKP equations.

\noindent {\bf Ke words:} {Soliton, Integrable equation, Hirota bilinear form, Lump solution}

 
\end{abstract}

\maketitle

\def \be {\begin{equation}}
\def \ee {\end{equation}}
\def \bea {\begin{eqnarray}}
\def \eea {\end{eqnarray}}
\def \ba {\begin{array}}
\def \ea {\end{array}}
\def \si {\sigma}
\def \al {\alpha}
\def \la {\lambda}
\def \D {\displaystyle }
\newcommand{\R}{\mathbb{R}}

\section{Introduction}
\setcounter{equation}{0}

The Korteweg-de Vries (KdV) equation and the Kadomtsev-Petviashvili (KP) equation
are nonlinear integrable differential equations, and their Hirota bilinear forms play a crucial role in generating their soliton solutions,
 a kind of exponentially localized solutions, describing diverse nonlinear phenomena (Hirota 2004).

By lump functions,
we mean analytical rational functions of spatial and temproal variables, which are localized in 
all directions in space. 
In recent years, there has been a growing interest
in lump function solutions (Berger and Milewski 2000; Gorshkov et al. 1993; Imai 1997;
Minzoni and Smyth 1996), called lump solutions
(see, e.g., Ablowitz and Clarkson 1991; Gilson and Nimmo 1990; Kaup 1981; Satsuma and Ablowitz 1979 
for typical examples).
 The KPI equation
\begin{equation}\label{eq:KPI3:ma-zhou-2015}
(u_t+6uu_x+u_{xxx})_x-3u_{yy}=0
\end{equation}
admits the following lump solution
\begin{equation}
u=4\frac {-[x+ay+3(a^2-b^2)t]^2+b^2(y+6at)^2+1/b^2}{\{[x+ay+3(a^2-b^2)t]^2
+b^2(y+6at)^2+1/b^2\}^2},\label{eq:lumpofKPI:ma-lump-jm-2015}
\end{equation}
where $a$ and $b\ne 0$ are free real constants (Manakov et al.1977).
Lump functions provide appropriate prototypes to model rogue wave dynamics
in both oceanography (Muller et al. 2005) and nonlinear optics (Solli et al. 2007).
There are various discussions on general rational function solutions to integrable equations
such as the KdV, KP, Boussinesq and Toda equations
(Ablowitz and Satsuma 1978; Adler and Moser 1978;
Ma et al. 2009; Ma and You 2004; Ma and You 2005).
It has become a very interesting topic to search for lump solutions or lump-type solutions, rationally localized solutions 
in almost all directions in space, to nonlinear partial differential equations, through the Hirota bilinear formulation.

In this paper, we would like to characterize 
positive quadratic functions and analyze
 positive quadratic function solutions to Hirota bilinear equations. Such
polynomial solutions generate 
lump or lump-type solutions to nonlinear partial differential equations under
the dependent variable 
transformations $u=2(\ln f)_x$ and $u=2(\ln f)_{xx}$, where $x$ is one of the spatial variables.
We will present sufficient and necessary conditions for positive quadratic functions to solve Hirota bilinear equations,
and apply the resulting theory to a few generalized KP and BKP equations.

\section{From Hirota bilinear equations to nonlinear equations}

 Let $M$ be a natural number and
$ x=(x_1,x_2,\cdots,x_M)^T$ in $\mathbb{R}^M$ be a column vector of independent variables.
For $f,g\in C^\infty ( \mathbb{R}^M)$, Hirota bilinear derivatives (Hirota 2004) are defined as follows:
\begin{equation}
D_1^{n_1} D_2^{n_2}
\cdots D_M^{n_M}
f\cdot g
:=\prod _{i=1}^M (\partial_{x_i}-\partial_{x_i'})^{n_i}
f(x)g(x')|_{x'=x},
\end{equation}
where $x'=(x'_1,x'_2,\cdots,x'_M)^T$ and $n_i\ge 0,\ 1\le i\le M$.
For example,
we have the first-order and second-order Hirota bilinear derivatives:
\be D_i f\cdot g=f_{x_i}g-fg_{x_i} , \   
D_i D_j f\cdot g=f_{x_i,x_j}g+fg_{x_i,x_j}-f_{x_i}g_{x_j}- f_{x_j}g_{x_i},
\ee
where $1\le i,j\le M$.

One basic property of the Hirota bilinear derivatives is that
\be D_{i_1}D_{i_2}\cdots D_{i_k} f\cdot g=(-1)^kD_{i_1}D_{i_2}\cdots D_{i_k} g\cdot f,\ee
where $1\le i_1,i_2,\cdots, i_k\le M$ need not be distinct.
It thus follows that if $k$ is odd, we have
\be D_{i_1}D_{i_2}\cdots D_{i_k} f\cdot f=0.
\ee
We will discuss the following general Hirota bilinear equation
\be\label{eq:GHEs:ma-zhou-2015}
 P(D)f\cdot f=P(D_1,D_2,\cdots,D_M)f\cdot f=0,
\ee
where $P$ is a polynomial of $M$ variables and $D= (D_{1}, D_{2}, \cdots$, $D_{M})$.
Since the terms of odd powers are all zeros, we assume that $P$ is an even polynomial, i.e., $P(-x)=P(x)$, and 
to generate non-zero polynomial solutions,
we require that $P$ has no 
constant term, i.e., $P(0)=0$.  
Moreover,
we set
\be P(x)=\sum_{i,j=1}^M
p_{ij}x_{i}x_{j}+\sum_{i,j,k,l=1}^Mp_{ijkl}
x_{i}x_{j}x_{k}x_{l}+\textrm{other terms},
\label{eq:defofP:ma-zhou-2015}
\ee
where $p_{ij} $ and $p_{ijkl}$ are coefficients of terms of second- and fourth-degree, to determine quadratic function solutions.

For convenience's sake, we adopt the index notation for partial derivatives of $f$:
\be f_{i_1i_2\cdots i_k}=\frac{\partial^k f}{\partial x_{i_1}\partial x_{i_2}\cdots\partial x_{i_k}},  \, \, 1\le i_1,i_2,\cdots, i_k\le M.
\ee
Using this notation, we have the compact expressions for the second- and fourth-order Hirota bilinear derivatives: 
 \be D_{i}D_{j}f\cdot f=2(f_{ij}f-f_{i}f_{j}),  \ 1\le i,j\le M,
\label{eq:formula1ofHirotaderivatives:ma-zhou-2015}
\ee
and
\begin{eqnarray}
 &&\quad D_{i}D_{j}D_{k}D_{l}(f\cdot f) \nonumber
\\
 &&= 2\bigl[f_{ijkl}f-f_{ijk}f_{l}-f_{ijl}f_{k}
-f_{ikl}f_{j}
\bigr.
\nonumber
\\
&&\quad
\bigl.
-f_{jkl}f_{i}+f_{ij}f_{kl}+f_{ik}f_{jl}+f_{il}f_{jk}\bigr],\ 1\le i,j,k,l\le M.
\label{eq:formula2ofHirotaderivatives:ma-zhou-2015}
\end{eqnarray}

Motivated by Bell polynomial theories on soliton equations
(Gilson et al. 1996; Ma 2013), we take
the dependent variable transformations:
\be
u=2(\ln f)_{x_1},\ u=2(\ln f)_{x_1x_1},\label{eq:Tranfs:ma-zhou-2015}
\ee
to formulate nonlinear differential equations from Hirota bilinear equations.
All integrable nonlinear equations can be generated this way (Chen 2006; Hirota 2004).

\begin{example}
  For the KdV equation 
  \be
  u_t+6uu_x+u_{xxx}=0 ,
  \ee
 the transformation $u=2(\ln f)_{xx}$ provides a link to the bilinear form
 \be
 (D_x D_t+D^4_x)f\cdot f=0.
  \ee
  For the KPI and KPII equations
  \be\label{eq:KPIsigma:ma-zhou-2015}
  (u_t+6uu_x+u_{xxx})_x+\sigma u_{yy}=0 ,\ \sigma=\mp 1,
  \ee
 the transformation $u=2(\ln f)_{xx}$ makes connection with the bilinear form
 \be
 (D_x D_t+D^4_x+\sigma  D_y^2)f\cdot f=0.
  \ee
\end{example}

If a polynomial solution $f$ is positive, then
the solution $u$ defined
by either of the  dependent variable transformations in (\ref{eq:Tranfs:ma-zhou-2015}) is analytical,
and most likely, rationally localized in space, and thus it often presents a lump solution to the corresponding nonlinear differential equation.
In what follows, we would like to analyze quadratic function solutions to Hirota bilinear equations to construct lump solutions
to nonlinear differential equations.

\section{Positive quadratic function solutions to bilinear equations}

\subsection{Non-negative and positive quadratic functions}

Let us consider a general quadratic function
\be
 f(x)=x^TAx-2b^Tx+ c , \  x\in \mathbb{R}^M,\label{eq:defoff:ma-zhou-2015}
\ee
where
$A\in \mathbb{R}^{M\times M}$ is a symmetric matrix,
 $b\in \mathbb{R}^M$ denotes a column vector, $ c \in \mathbb{R}$ is a constant and $T$ denotes transpose.

We say that a polynomial $f$  is non-negative (or positive) if $f(x)\ge 0,\  \forall x\in \mathbb{R}^M$
(or $f(x)> 0,\  \forall x\in \mathbb{R}^M$).
We need the pseudoinverse of a matrix
to determine the non-negativity (or positivity) of a quadratic function.

For a matrix $A\in \mathbb{R}^{N\times M}$, we call a matrix
$A^+\in\mathbb{R}^{M\times N}$ the Moore-Penrose pseudoinverse  of  $A$ if
\be AA^+A=A,\ A^+AA^+=A^+,\  (AA^+)^T=AA^+,\ (A^+A)^T=A^+A,
\label{eq:defofA^+:ma-zhou-2015}
\ee
which uniquely defines $A^+$ 
for any given matrix $A$ (Penrose 1955).
Obviously, the Moore-Penrose pseudoinverse of a zero matrix is the zero matrix itself, and $(A^+)^T=(A^T)^+$, which implies that
 if $A$ is symmetric, then so is $A^+$.
When a square matrix $A$ is non-singular, i.e., $|A|=\det (A)\ne 0$, we have $A^+=A^{-1}$, $A^{-1}$ being the inverse of $A$.

Suppose that a non-zero matrix $A\in \mathbb{R}^{N\times M}$ has its
singular value decomposition
\be
 A=U \left[\begin{array} {cc} \Sigma &0 \vspace{2mm}\\
0&0 \end{array} \right]
V^T,
 \ee
 where $U\in \mathbb{R}^{N\times N}$ and $V\in \mathbb{R}^{M\times M} $ are orthogonal matrices, and $\Sigma$ reads 
\be
 \Sigma=\textrm{diag}(d_1,\cdots,d_r) ,\   d_1\ge\cdots\ge d_r>0,\ r=\textrm{rank}(A).
\ee
 Then
the Moore-Penrose pseudoinverse of $A$ is given by 
\be
A^+=V
\left[\begin{array} {cc} \Sigma ^{-1} &0 \vspace{2mm}\\
0&0 \end{array} \right]
U^T.\ee

The Moore-Penrose pseudoinverse can be 
applied to
analysis of linear systems (Penrose 1955).
A linear system $A\alpha =b$ 
is consistent if and only if 
$AA^+b=b$.
Moreover, if it is consistent, then 
its solution set 
is given by
\[
\{\alpha = A^+b +(I_M-A^+A)\beta \, | \, \beta \in \mathbb{R}^{M}\},
\]
where $I_M$ is the identity matrix of size $M$.

\begin{lemma}
\label{lem:basicpropertiesofA^+:ma-zhou-2015}
Let $A\in \mathbb{R}^{M\times M} $ be symmetric and $b \in \mathbb{R}^{M}$ be arbitrary.
If $\alpha \in \mathbb{R}^{M}$ solves $A\alpha =b$, then
\be \alpha ^T A\alpha = b^TA^+b ,
\label{eq:aninvariantofls:ma-zhou-2015}
\ee
and further,
\be  (x-\alpha )^TA(x-\alpha ) = (x-A^+b)^TA(x-A^+b). \label{eq:thesameresultforallsols:ma-zhou-2015}\ee
\end{lemma}

\begin{proof}
Recalling the first property in (\ref{eq:defofA^+:ma-zhou-2015}) and  
using 
$A^T=A$,
we have
\[
\alpha ^TA\alpha =\alpha ^TAA^+A\alpha
=\alpha ^T A^T A^+A\alpha 
= b ^TA^+b.
\]
Therefore,  
(\ref{eq:aninvariantofls:ma-zhou-2015}) holds.
Now, noting that
\[\alpha ^TAx=(A\alpha )^T x =b^Tx,\  x^T (A \alpha) =x^T b =b^Tx,\]
 \[
(A^+b)^TAx=b^TA^+Ax=\alpha ^T A A^+Ax=\alpha ^TAx= (A\alpha )^Tx=b^Tx,
\]
and \[
x^TA(A^+b)=x ^T A A^+A\alpha =x ^TA\alpha = x^T b=b^Tx,
\]
we see that (\ref{eq:thesameresultforallsols:ma-zhou-2015}) follows directly from (\ref{eq:aninvariantofls:ma-zhou-2015}).
 \end{proof}

We denote a positive-semidefinite (or positive-definite) matrix $A\in \mathbb{R}^{M\times M}$ by $A\ge 0$ (or $A> 0$). Namely,
$A\ge 0$ (or $A>0$) means that $x^TAx\ge 0$ for all $x\in   \mathbb{R}^M$
(or $x^TAx>0$ for all non-zero $x\in   \mathbb{R}^M$).
The following theorem gives a description of non-negative (or positive) quadratic functions.

\begin{theorem}
\label{thm:positiveness:ma-zhou-2015}
Let a quadratic function 
$f$ be defined by (\ref{eq:defoff:ma-zhou-2015}).
Then 
 (a) 
 if $b \in \textrm{range}(A)$, then 
\begin{eqnarray}&&
 f(x)=
(x-\alpha)^TA(x-\alpha)+ c -\alpha^TA\alpha
\nonumber 
\\
&& \qquad =
(x-A^+b)^TA(x-A^+b)+ c -b^TA^+b,
\label{eq:PQFs:ma-zhou-2015} \end{eqnarray}
where $\alpha \in \mathbb{R}^M$ solves $A\alpha =b$;
and 
(b) $f$ is non-negative (or positive) if and only if
 $A\ge 0$, $b \in \textrm{range}(A)$ and
\be d = c -b^TA^+b
\label{eq:formulaforc:ma-zhou-2015}
\ee
is non-negative (or positive).
\end{theorem}

\begin{proof} 
(a) First, based on  Lemma
\ref{lem:basicpropertiesofA^+:ma-zhou-2015},
it is sufficient to show that
\be 
 f(x)=x^TAx-2\alpha ^TAx+ c =(x-\alpha)^TA(x-\alpha)+ c -\alpha^TA\alpha,
\label{eq:formulaforf(x):ma-zhou-2015}
\ee
where we have made use of $b=A\alpha $ and $A^T=A$.

(b) Second, we prove part (b). 

($\Leftarrow$) 
This directly follows from the second equality of 
\eqref{eq:PQFs:ma-zhou-2015} in
part (a).

($\Rightarrow$)  Suppose that $A\ge 0$ is false. Then there exists a vector $\beta\in \mathbb{R}^M$ such that $\beta^TA\beta<0$, and further  for $r \in \mathbb{R}$, we have
$$f(r\beta)=r^2\beta^TA\beta-2rb^T\beta+ c  \to -\infty,\  \mathrm{as}\   r\to\pm \infty.$$
This is a contradiction to the assumption on $f$ that $f$ is non-negative (or positive). Therefore, we have $A\ge 0$.

Now
let $b=b^{(1)}+b^{(2)}$ with $b^{(1)}\in \mathrm{range}(A)$ and $b^{(2)}\in \mathrm{range}(A)^\bot$.
Assume that $ \alpha\in \mathbb{R}^M$ satisfies $A\alpha=b^{(1)}$. Consider $x=\alpha+r b^{(2)}$, with $r$ being a positive number. Then we can have
$$
\begin{array}{l}
  f(x)=x^TAx-2\alpha^TAx-2b^{(2)T}x+ c  \vspace{1.5mm}
\\
  \qquad \  =(x-\alpha)^TA(x-\alpha)-2b^{(2)T}x+ c -\alpha ^T A \alpha
\vspace{1.5mm}
\\
  \qquad \
=r^2b^{(2)T}Ab^{(2)}-2b^{(2)T}\alpha-2rb^{(2)T}b^{(2)}+ c -\alpha ^T A \alpha
\vspace{1.5mm}
\\
  \qquad \
=-2r b^{(2)T}b^{(2)}-2b^{(2)T}\alpha+ c -\alpha ^T A \alpha
  \to -\infty,\  \mathrm{as}\  r\to \infty,
\end{array}$$
if $b^{(2)}\ne 0$.
Therefore $b^{(2)}=0$, since $f$ is non-negative (or positive). This implies $b\in \mathrm{range}(A)$.
Further, $ d =f(\alpha)\ge 0$ (or $>0$).
 The proof is finished.
 \end{proof}

Any two solutions $\alpha^{(1)}$ and $\alpha^{(2)}$ to $A\alpha=b$ satisfy $A(\alpha^{(1)}-\alpha^{(2)})=A\alpha^{(1)}-A\alpha^{(2)}=0$, 
which means that $\alpha^{(1)}-\alpha^{(2)}\in$ ker$(A)$ and thus
\[\alpha ^{(1)T}A\alpha ^{(1)}=
\alpha ^{(1)T}A\alpha ^{(2)}=\alpha ^{(2)T}A\alpha ^{(1)}
=\alpha ^{(2)T}A\alpha ^{(2)}.
\]
This is just a consequence of (\ref{eq:aninvariantofls:ma-zhou-2015}).
We also point out that all the results presented in an earlier paper (Jankovic 2005) are consequences of our results in Theorem \ref{thm:positiveness:ma-zhou-2015}.
For example, we can have 
the last two thorems in (Jankovic 2005), i.e.,
Theorems 6 and 7 in (Jankovic 2005): 
a quadratic function $f$ is bounded from below 
if and only if $f$
reaches its minimum at a point $x_0\in \mathbb{R}^M$ if and only if $A\ge 0$ and $Ax_0=b$, where $f$ is assumed to 
be given by (\ref{eq:defoff:ma-zhou-2015}).
Actually, Theorem \ref{thm:positiveness:ma-zhou-2015} also tells that $f$ achieves its minimum at any point $\alpha \in \mathbb{R}^M$, where $\alpha $ is a solution 
to $A\alpha =b$, and its minimum is 
$c-b^TA^+b$. We proves the result on the extreme value as follows. 

{\corollary 
If a quadratic function defined by (\ref{eq:defoff:ma-zhou-2015}) reaches its minimum or maximum, then its extreme value is $c-b^TA^+b$.
}
\begin{proof}
If $f$ reaches its maximum $\gamma$, then $g=f-\gamma  $ is non-negative and by Theorem 
\ref{thm:positiveness:ma-zhou-2015}, we have
 $g(x)=(x-\alpha )^TA(x-\alpha)+ (c-\gamma )-b^TA^+b$ with $A\ge 0$, which says that 
$f(x) \ge   c-b^TA^+b$ and $f(\alpha )= c-b^TA^+b$, and so the minimum value of  $f$ is $ c-b^TA^+b$. 
If $f$ reaches its maximum, then $g=-f$ reaches its minimum. Therefore, as we just proved, $g$ achieves the minimum value $-c+b^TA^+b$, and so 
 $f$ has the maximum value  $ c-b^TA^+b$, which completes the proof. 
\end{proof}

\subsection{Positive quadratic function solutions}

Let $\alpha=(\alpha _1,\cdots,\alpha_M)^T
\in \mathbb{R}^M$ be a fixed vector. Consider
 a quadratic function defined as follows:
\be\label{eq:concretePQFs:ma-zhou-2015}
f(x)=(x-\alpha)^TA(x-\alpha)+
d
=\sum_{i,j=1}^Ma_{ij}{(x_i-\alpha_i)}{(x_j-\alpha_j)}+
d ,
\ee
where the real matrix $A=(a_{ij})_{M\times M}$ is symmetric and $ d\in \mathbb{R}$ is a constant.
Theorem \ref{thm:positiveness:ma-zhou-2015}
 guarantees that when $A\ge 0$ and $d>0$, this presents the class of positive quadratic functions.

Obviously, we have
\[D_{i_1}D_{i_2}\cdots D_{i_k} f\cdot f=0, \ 1\le i_j\le M, \ 1\le j\le k, \ k>4,\]
 for any quadratic function
$f$. 
Moreover, because all odd-order Hirota bilinear derivative terms in 
the Hirota bilinear equation
(\ref{eq:GHEs:ma-zhou-2015}) are zero,
the bilinear equation (\ref{eq:GHEs:ma-zhou-2015}) is reduced to
\be\label{eq:reducedHBEs:ma-zhou-2015}
 Q(D)f\cdot  f=0,
\ee
where
\be Q(x)=\sum_{i,j=1}^Mp_{ij}x_{i}x_{j}+\sum_{i,j,k,l=1}^Mp_{ijkl}
x_{i}x_{j}x_{k}x_{l},
\ee
since $Q(D)f\cdot f=P(D)f\cdot f$ for $P$ defined by (\ref{eq:defofP:ma-zhou-2015}).

Now we compute the second- and fourth-order Hirota bilinear derivatives of a
positive quadratic function defined by (\ref{eq:concretePQFs:ma-zhou-2015}).
Note that
 \[
f_{i}=2\sum_{k=1}^M a_{ik}(x_k-\alpha_k)=2A_i^T(x-\alpha),\
f_{ij}=2a_{ij},\ 1\le i,j\le M,
\]
 where $A_i$ is the $i$th column vector of $A$ for $1\le i\le M$.
We denote $y=x-\alpha.$ Then using
(\ref{eq:formula1ofHirotaderivatives:ma-zhou-2015}), we have
 \begin{eqnarray}
 &&
\sum_{i,j=1}^Mp_{ij}D_{i}D_{j}f\cdot f
=4\sum_{i,j=1}^Mp_{ij}a_{ij}f-8\sum_{i,j=1}^M p_{ij} y^TA_iA_j^T y
\nonumber \\
&&
 =4 d \sum_{i,j=1}^M p_{ij}a_{ij}+4y^T\Bigr[\sum_{i,j=1}^M p_{ij}(a_{ij}A-A_iA_j^T-A_jA_i^T)\Bigl]y.
 \end{eqnarray}
By (\ref{eq:formula2ofHirotaderivatives:ma-zhou-2015}), the fourth-order Hirota bilinear derivatives of $f$ in (\ref{eq:concretePQFs:ma-zhou-2015}) read
\be
 D_{i}D_{j}D_{k}D_{l}f\cdot f=2(f_{ij}f_{kl}
+f_{ik}f_{jl}+f_{il}f_{jk})
=8(a_{ij}a_{kl}+a_{ik}a_{jl}+a_{il}a_{jk}).
  \ee
Thus, if (\ref{eq:concretePQFs:ma-zhou-2015}) solves
the Hirota bilinear equation
(\ref{eq:GHEs:ma-zhou-2015}), i.e.,
the reduced Hirota bilinear equation
(\ref{eq:reducedHBEs:ma-zhou-2015}), then we have
\begin{eqnarray}
&&  8\sum_{i,j,k,l=1}^Mp_{ijkl}(a_{ij}a_{kl}+a_{ik}a_{jl}+a_{il}a_{jk})
+4 d \sum_{i,j=1}^Mp_{ij}a_{ij}
\nonumber
\\
&& +y^T\Bigl[\sum_{i,j=1}^Mp_{ij}(a_{ij}A-A_iA_j^T-A_jA_i^T)\Bigr]y=0.
\end{eqnarray}
Note $x\in \mathbb{R}^M$ is arbitrary, and so is $y=x-\alpha $. Therefore, we obtain the following result.

 \begin{theorem}
\label{thm:condsofpositivesols:ma-zhou-2015}
Let $A=(a_{ij})_{M\times M}\in \mathbb{R}^{M\times M}$ be symmetric and $d\in \mathbb{R}$ be arbitrary.  
A quadratic function $f$ defined by (\ref{eq:concretePQFs:ma-zhou-2015}) solves the Hirota bilinear equation
(\ref{eq:GHEs:ma-zhou-2015})
if and only if
   \be\label{eq:c1forpolysol:ma-zhou-2015}
   \displaystyle 2\sum_{i,j,k,l=1}^Mp_{ijkl}(a_{ij}a_{kl}+a_{ik}a_{jl}+a_{il}a_{jk})
+ d \sum_{i,j=1}^Mp_{ij}a_{ij}=0
\ee
   and
\be
\label{eq:c2forpolysol:ma-zhou-2015}
\sum_{i,j=1}^Mp_{ij}(a_{ij}A-A_iA_j^T-A_jA_i^T)=0,
\ee
where $A_i$ denotes the $i$th column vector of the symmetric matrix $A$ for $1\le i\le M$.
 \end{theorem}

{\corollary \label{cor:constant:ma-zhou-2015}If $f(x)=x^TAx+ d $ solves the Hirota bilinear equation (\ref{eq:GHEs:ma-zhou-2015}), then for any $\a\in \mathbb{R}^M$,
$f(x-\a)$ solves the Hirota bilinear equation (\ref{eq:GHEs:ma-zhou-2015}), too.
}
\begin{proof} This is because (\ref{eq:c1forpolysol:ma-zhou-2015}) and (\ref{eq:c2forpolysol:ma-zhou-2015}) 
only depend on the matrix $A$ and the constant $ d $, but do not depend on the shift vector $\alpha$.
\end{proof}

We denote
 the coefficient matrix of the second order Hirota bilinear derivative terms by
 \be 
\label{eq:defofP^{(2)}:ma-zhou-2015} 
P^{(2)}=(p_{ij})_{M\times M}\in \mathbb{R}^{M\times M},\ee
in the Hirota bilinear equation (\ref{eq:GHEs:ma-zhou-2015}).
When $P^{(2)}=0$, the matrix equation \eqref{eq:c2forpolysol:ma-zhou-2015}
is automatically satisfied and the scalar equation \eqref{eq:c1forpolysol:ma-zhou-2015} reduces to
   \begin{equation}
\label{eq:cond4forpolysol:ma-zhou-2015}
   \displaystyle \sum_{i,j,k,l=1}^Mp_{ijkl}(a_{ij}a_{kl}+a_{ik}a_{jl}+a_{il}a_{jk})
=0.
\end{equation}
If $M\ge 2 $, for a fixed matrix $A$, 
obviously there exists infinitely many non-zero solutions of $p_{ijkl},\  1\le i,j,k,l\le M,$ to the equation \eqref{eq:cond4forpolysol:ma-zhou-2015}.

Let us now consider quadratic function solutions with $|A|\ne 0$. 

If $M=1$, then $a_{11}\ne 0$. Therefore,  
\eqref{eq:c1forpolysol:ma-zhou-2015}  and \eqref{eq:c2forpolysol:ma-zhou-2015} equivalently yield
\[ p_{11}=p_{1111}=0.\]
This means that 
a bilinear ordinary differential equation defined by 
\eqref{eq:GHEs:ma-zhou-2015} has a quadractic function solution if and only if 
the least degree of a polynomial $P$ must be greater than 5. 

If $M=2$, 
we have th following example in (1+1)-dimensions.
Consider the function
$f(x,t)= 3x^2-2xt+t^2+\frac {27}2,  $
where $A=\left[ \begin{array}{cc} 3 &-1\\ -1 & 1 \end{array}\right]$ with $|A|=2>0$.
Obviously, this quadratic polynomial is positive,
and
 solves the following (1+1)-dimensional Hirota bilinear equation:
$$(D_x^4-D_x^2-2D_tD_x-3D_t^2)f\cdot f=0,$$
where the symmetric coefficient matrix $P^{(2)}=
\left[\begin{array}{cc} -1 & -1 \vspace{0mm}\\
-1 & -3 \end{array} \right] 
$ is not zero. This function $f$ leads to lump solutions to 
the corresponding nonlinear equations
under $u=2(\ln f)_x$ or $u=2(\ln f)_{xx}$.   

When $M\ge 3$, there is a totally different situation.
What kind of Hirota bilinear equations
(\ref{eq:GHEs:ma-zhou-2015})
can possess a quadratic function solution 
defined by (\ref{eq:concretePQFs:ma-zhou-2015}) 
with $|A|\ne 0$?
The following theorem provides a complete
answer to this question. 

 \begin{theorem}
\label{thm:condsofAnonsingular:ma-zhou-2015}
Let $M\ge 3$.
Assume that a quadratic function $f$ defined by
\eqref{eq:concretePQFs:ma-zhou-2015} 
solves the Hirota bilinear equation
\eqref{eq:GHEs:ma-zhou-2015} with $P$ defined by  
\eqref{eq:defofP:ma-zhou-2015}.
If $|A|\ne 0$, i.e., $A$ is non-singular, 
then \be p_{ij}+p_{ji}=0,\ 1\le i,j\le M,\ee
which means that the Hirota bilinear equation
\eqref{eq:GHEs:ma-zhou-2015} doesn't contain any second-order Hirota bilinear derivative term.
\end{theorem}
 
\begin{proof}
First, assume that $P^{(2)T}=P^{(2)}$. 
Then,
\eqref{eq:c2forpolysol:ma-zhou-2015} becomes
\begin{equation}
\label{eq:cond3forpolysol:ma-zhou-2015}
\tilde {a} A-2AP^{(2)}A=0,\ \textrm{where}\ \tilde {a} =\sum_{i,j=1}^Mp_{ij}a_{ij}.
\end{equation}
Since $A$ is symmetric, 
there exists an orthogonal matrix $U\in \mathbb{R
}^{M\times M} $ such that 
\[\hat{A}= U^TAU=\diag (\hat {a}_1,\cdots, \hat{a}_M).\]
Set $
\hat{P}^{(2)}=U^TP^{(2)}U,$
and by \eqref{eq:cond3forpolysol:ma-zhou-2015},  we have
\be  {\tilde {a}} \hat{A}-2\hat{A}\hat{P}^{(2)}\hat{A}=0.
\label{eq:cond5forpolysol:ma-zhou-2015}
\ee
Since $|A|\ne 0$, we have $|\hat {A}|\ne 0$. Thus, \eqref{eq:cond5forpolysol:ma-zhou-2015} tells that 
$ \hat{P}^{(2)}=\frac{{\tilde{a}} }{2}\hat{A}^{-1}$ and further
$\hat {P}^{(2)}$ is diagonal. Therefore, we can express
\[ \hat{P}^{(2)}=\diag(\hat{p}_1,\cdots,\hat{p}_M).\]
Plugging the two diagonal matrices 
$\hat A$ and $\hat {P}^{(2)}$ 
 into \eqref{eq:cond5forpolysol:ma-zhou-2015}
  engenders
 \be
\label{eq:property1foralpha:ma-zhou-2015}
{\tilde{a}}=2\hat{a}_k\hat{p}_k, \ 1\le k\le  M.
\ee

On the other hand,
 a direct calculation can show that
$\tilde a=\sum_{i,j=1}^Ma_{ij}p_{ij}$ is an invariant under an orthogonal similarity transformation,
and thus, from $\hat A=U^TAU$ and $\hat {P}^{(2)}=U^TP^{(2)}U$, we have
 \be {\tilde{a}}=\sum_{k=1}^M\hat{a}_k\hat{p}_k.
\label{eq:property2foralpha:ma-zhou-2015}
\ee 

Now a combination of \eqref{eq:property1foralpha:ma-zhou-2015} and 
\eqref{eq:property2foralpha:ma-zhou-2015}
 tells that $ M{\tilde{a}}=2{\tilde{a}}$. 
Since $M\ge 3$, we see $ {\tilde{a}}=0$, and so,
 $\hat {P}^{(2)}=0$,
which implies that 
  $P^{(2)}=0$.

Second, if $P^{(2)}$ is not symmetric, noting that
\[ 
\sum_{i,j=1}^N p_{ij}x_ix_j = 
\sum_{i,j=1}^N \bar p_{ij}x_ix_j,
\ 
\bar p_{ij}=\frac {p_{ij}+p_{ji}}2,\ 1\le i,j\le M.
\]
we can begin with a symmetric coefficient matrix of second order Hirota bilinear derivative terms, 
 $\bar P^{(2)}=(\bar p_{ij})_{M\times M}$, to analyze quadratic function solutions. Thus,
as we just showed, 
$\bar P^{(2)}=0$.
This is exactly what we need to get.
The proof is finished.
\end{proof}

Theorem \ref{thm:condsofAnonsingular:ma-zhou-2015}
tells us about the case of $|A|\ne 0$, which says that 
if a Hirota bilinear equation admits a quadratic function solution determined by \eqref{eq:concretePQFs:ma-zhou-2015} with $|A|\ne 0$, then it
cannot contain 
any second-order Hirota bilinear derivative term.   

For the KPI and KPII equations, since 
the corresponding symmetric coefficient matrix $P^{(2)}$ is not zero, 
Theorem 
\ref{thm:condsofAnonsingular:ma-zhou-2015}
tells that any quadratic function solution $f$ cannot be expressed as
 a sum of squares of three linear functions and a constant: $f=g_1^2+g_2^2+g_3^2+d$, where 
\[ g_i= c_{i1}x+c_{i2}y+c_{i3} t+c_{i4},\ 1\le i\le 3,
\]
with $(c_{ij})_{3\times 3}$ being non-singular, which will also be showed clearly later. 

The other case is $|A|=0$, for which 
there is no requirement on inclusion of second-order Hirota bilinear derivative terms.
Obviously, when
$A=\diag(a_1,\cdots,a_{M-1},0) \ne 0$,
 \eqref{eq:cond3forpolysol:ma-zhou-2015} has 
a non-zero symmetric matrix solution 
$P^{(2)}=\diag(\underbrace{0,\cdots,0}_{M-1},1)\ne 0$ with $\tilde a=0$,
 and \eqref{eq:c1forpolysol:ma-zhou-2015} has infinitely many non-zero solutions for $\{p_{ijkl}| \, 1\le i,j,k,l,\le M\}$. 
Therefore, we can have both second- and fourth-order Hirota bilinear derivative terms in
the Hirota bilinear equation
\eqref{eq:GHEs:ma-zhou-2015}. 

\subsection{Solutions as sums of squares of linear functions}

We will explore relations between quadratic function solutions
and sums of squares of linear functions, and discuss quadratic function solutions which can be written as sums of squares of linear functions.

{\theorem
\label{thm:PQFasSumofSquaresofLFs:ma-zhou-2015}
Let a quadratic function $f$ be defined by
\eqref{eq:concretePQFs:ma-zhou-2015}. Suppose $r=\textrm{rank}(A)$. Then there exist $b^{(j)}\in \mathbb{R}^M$, $c_j\in \mathbb{R}$, $1\le j\le r$, such that 
\be 
\label{eq:sumformulationforPQF:ma-zhou-2015}
f(x) =\sum_{j=1}^r (b^{(j)T}x+c_j)^2+d.
\ee
}
\begin{proof}
We assume 
that the symmetric matrix $A$ has 
the singular value decomposition: 
\be
 A=V \left[\begin{array}{cc} \Sigma &0 \vspace{2mm}\\
0& 0\end{array}\right] 
V^T,
 \ee
 where $V\in \mathbb{R}^{M\times M}$ is orthogonal and 
\[
 \Sigma=\textrm{diag}(d_1,\cdots,d_r) ,\   d_1\ge\cdots\ge d_r>0.
\]
Upon denoting $V = (v^{(1)},v^{(2)},\cdots,v^{(M)})$ and setting
\be  {b}^{(j)}=\sqrt{d_j}\, v^{(j)},\  {c}_j=-\alpha^T {b}^{(j)},\  1\le j\le r,\ee 
we have 
$$
A=\sum_{j=1}^rd_jv^{(j)}v^{(j)T}=\sum_{j=1}^r(\sqrt{d_j}\,v^{(j)})(\sqrt{d_j}\,v^{(j)})^T
=\sum_{j=1}^r {b}^{(j)} {b}^{(j)T},
$$
and thus
\bea
f(x)&=&\sum_{j=1}^r(x-\alpha)^T {b}^{(j)} {b}^{(j)T}(x-\alpha)+d \nonumber\\
&=&\sum_{j=1}^r[(x-\alpha)^T{b}^{(j)}][(x-\alpha)^T{b}^{(j)}]^T+d\nonumber\\
&=&\sum_{j=1}^r({b}^{(j)T}x+{c}_j)^2+d.\nonumber
\eea
The proof is finished.
\end{proof}

Based on Theorem \ref{thm:positiveness:ma-zhou-2015}, and noting that constant functions are particular linear functions, the following result is a direct consequence of Theorem \ref{thm:PQFasSumofSquaresofLFs:ma-zhou-2015}.

{\corollary
Any non-negative quadratic function can be written as a sum of squares of linear functions.
}

This corollary guarantees that completing squares can transform non-negative quadratic functions into sums of 
squares of linear functions. It also proves Hilbert's 17th problem for quadratic functions. 

{\lemma \label{lem:consistency:ma-zhou-2015}
Let $N$ be a natural number, and $b^{(j)}\in \mathbb{R}^M,\  c_j\in \mathbb{R},\ 1\le j\le N$, be arbitrary.
Then the linear system
\be 
\Bigl(\sum_{j=1}^N b^{(j)}b^{(j)T}\Bigr)\alpha =-\sum_{j=1}^N c_jb^{(j)},
\label{eq:Aspeciallinearsystem:ma-zhou-2015}
\ee 
is consistent, where $\alpha \in \mathbb{R}^M$ an unknown vector.
}
\begin{proof}
Note that
the columns of the coefficient matrix $\sum_{j=1}^N b^{(j)}b^{(j)T}$
read
\[\sum_{j=1}^N b_{1}^{(j)}b^{(j)},\cdots,
\sum_{j=1}^N b_{M}^{(j)}b^{(j)}
 ,\]
where
$b^{(j)}_i$ is the $i$th component of $b^{(j)}$.
It follows that
the dimension of the column space of the coefficient matrix is equal to
the rank of
the $M\times N$ matrix $(b^{(j)}_i)_{1\le i\le M,1\le j\le N}$.
This implies that
 the column space of the coefficient matrix
is just the space spanned by $b^{(j)},\ 1\le j\le N$.
On the other hand, the given vector $-\sum_{j=1}^N c_jb^{(j)}$ belongs to the space spanned by
$b^{(j)},\ 1\le j\le N$. Therefore, the linear system is consistent.
\end{proof}

{\theorem
\label{thm:solofsumofsqures:ma-zhou-2015}
Let $N$ be a natural number, and $b^{(j)}\in \mathbb{R}^M,\  c_j\in \mathbb{R},\ 1\le j\le N$, $h\in \mathbb{R}$ be arbitrary. 
Suppose that a quadratic function $f$ is given by
\be
f(x)=\sum_{j=1}^N(b^{(j)T}x+c_j)^2+h.
\label{eq:fofsumofsqurefunctionswithh:ma-zhou-2015}
\ee
Then 
(a) 
we have 
\be 
\label{eq:positivefunctionform:ma-zhou-2015}
f(x)=
x^TAx-2b^Tx+c=
(x-A^+b )^TA(x-A^+b )+d,
\ee
where 
\be\label{eq:defofAbcdandcinsumofsqurefunction:ma-zhou-2015}
A=\sum_{j=1}^Nb^{(j)}b^{(j)T},\ 
b=-\sum_{j=1}^N c_jb^{(j)},\ 
c=\sum_{j=1}^N c^2_j+h,\ 
d = c
 -b^TA^+b;
\ee
(b)
$f$ is non-negative (or positive) if and only if $ d \ge 0 $ (or $ d >0$); and
(c) $f$ solves the Hirota bilinear equation
(\ref{eq:GHEs:ma-zhou-2015}), or equivalently
(\ref{eq:reducedHBEs:ma-zhou-2015}), if and only if
(\ref{eq:c1forpolysol:ma-zhou-2015}) and (\ref{eq:c2forpolysol:ma-zhou-2015}) are true for
the matrix $A$ and the constant $ d $ defined in (\ref{eq:defofAbcdandcinsumofsqurefunction:ma-zhou-2015}).
}
\begin{proof}
To prove part (a), we begin by computing that
\begin{eqnarray}
&& f(x)\di =\sum_{j=1}^N (b^{(j)T}x+c_j)^2+h
\nonumber
\\
&& \qquad \di =\sum_{j=1}^N (b^{(j)T}x)^T(b^{(j)T}x)+2\sum_{j=1}^N (b^{(j)T}x)c_j+\sum_{j=1}^N c^2_j+h
\nonumber
\\
&&\qquad \di = x^T \Bigl(\sum_{j=1}^N b^{(j)}b^{(j)T}\Bigr)x+2\Bigl(\sum_{j=1}^Nc_jb^{(j)T}\Bigr) x+\sum_{j=1}^N c^2_j+h\nonumber \\
&& \qquad
=x^TAx-2b^Tx+c 
,\label{eq:expofsumofsqures:ma-zhou-2015}
\end{eqnarray}
where $A,b$ and $c$ are defined in (\ref{eq:defofAbcdandcinsumofsqurefunction:ma-zhou-2015}).
It then follows from Lemma \ref{lem:consistency:ma-zhou-2015} and Theorem \ref{thm:positiveness:ma-zhou-2015} 
 that
\[
f(x)=(x-\alpha )^TA(x-\alpha )+d=(x-A^+b )^TA(x-A^+b)+d,
\]
where $\alpha$ solves $A\alpha =b$ and $d$ is defined in (\ref{eq:defofAbcdandcinsumofsqurefunction:ma-zhou-2015}). Therefore, part (a) is true.

Now, based on part (a) and noting that $A$ is positive-semidefinite,
 parts (b) and (c) are just consequences of Theorem \ref{thm:positiveness:ma-zhou-2015} and Theorem \ref{thm:condsofpositivesols:ma-zhou-2015}.
The proof is finished.
\end{proof}

This theorem tells
us 
the way of 
constructing positive quadratic function solutions through taking sums of squares of linear functions. It also leads to the following inequality involving
the Moore-Penrose pseudoinverse.

{\corollary
Let $N$ be a natural number, and $b^{(j)}\in \mathbb{R}^M,\  c_j\in \mathbb{R},\ 1\le j\le N$, be arbitrary.
Then \begin{equation}
\Bigl(\sum_{j=1}^Nc_jb^{(j)T}\Bigr)
A^+
\Bigl(\sum_{j=1}^Nc_jb^{(j)}\Bigr)
\le \sum_{j=1}^N c^2_j,\end{equation}
where $A^+$ is the
the Moore-Penrose pseudoinverse of $A=\sum_{j=1}^Nb^{(j)}b^{(j)T}$.
}

\begin{proof}
In Theorem \ref{thm:solofsumofsqures:ma-zhou-2015},
we assume that $h\ge 0$, and then
the quadratic function $f$ defined by (\ref{eq:fofsumofsqurefunctionswithh:ma-zhou-2015}) is non-negative, which means that
\[\sum_{j=1}^N c^2_j +h -b^TA^+b \ge 0 ,\]
where $b=-\sum_{j=1}^Nc_jb^{(j)}$.
The required result in the corollary follows immediately from taking a limit of the above inequality as $h\to 0$.
 \end{proof}

If the linear system (\ref{eq:Aspeciallinearsystem:ma-zhou-2015}) has a
particular solution $\alpha\in \mathbb{R}^M$ determined by
\[ b^{(j)T}\alpha =-c_j,\ 1\le j\le N,\]
 then we have
\[
\Bigl(\sum_{j=1}^Nc_jb^{(j)T}\Bigr)
A^+
\Bigl(\sum_{j=1}^Nc_jb^{(j)}\Bigr)
 =\sum_{j=1}^N c_j^2.\]
This is because by \eqref{eq:aninvariantofls:ma-zhou-2015}, we can compute that 
\[\begin{array}{l}
\displaystyle   \quad \Bigl(\sum_{j=1}^Nc_jb^{(j)T}\Bigr)
A^+
\Bigl(\sum_{j=1}^Nc_jb^{(j)}\Bigr)= \alpha ^TA\alpha 
\vspace{2mm}\\
\displaystyle 
= \alpha ^T\Bigl( \sum_{j=1}^N b^{(j)}b^{(j)T}\Bigr) \alpha
 =-\alpha ^T \sum_{j=1}^N b^{(j)}c_j=\sum_{j=1}^Nc_j^2.
\end{array}
\]

Next, we are going to present a basic characteristic of sums of squares of linear functions.

{\theorem
\label{cor:sumofsquare:ma-zhou-2015}
Let $N$ be a natural number, and $b^{(j)}\in \mathbb{R}^M$, $ c_j\in \mathbb{R},\ 1\le j\le N$, $h\in \mathbb{R}$ be arbitrary. Suppose that a quadratic function $f$ is defined by
\eqref{eq:fofsumofsqurefunctionswithh:ma-zhou-2015}, 
i.e., 
\[
f(x)=\sum_{j=1}^N(b^{(j)T}x+c_j)^2+h,
\]
and set
$  A=\sum_{j=1}^Nb^{(j)}b^{(j)T}$ and $ r=rank(A).
$
Then 
(a)
there exist $\tilde{b}^{(j)}\in \mathbb{R}^M$, $ \tilde{c}_j\in \mathbb{R},\ 1\le j\le r$, such that
\be
 f(x)=\sum_{j=1}^r(\tilde{b}^{(j)T}x+\tilde{c}_j)^2+\sum_{j=1}^N c^2_j
+h-b^TA^+b,\label{eq:fofsumofsqurefunctionsim:ma-zhou-2015}
\ee
where $
b=-\sum_{j=1}^N c_j{b}^{(j)};$
(b) if
$\di
f(x)=\sum_{j=1}^s(\hat{b}^{(j)T}x+\hat{c}_j)^2+\hat h$, where $\hat{b}^{(j)}\in \mathbb{R}^M$, $ \hat{c}_j\in \mathbb{R},\ 1\le j\le s$, $\hat h \in \mathbb{R}$, then
$s\ge r.$
}

\begin{proof}
(a) A combination of Theorem 
\ref{thm:PQFasSumofSquaresofLFs:ma-zhou-2015} and Theorem \ref{thm:solofsumofsqures:ma-zhou-2015}
leads to part (a).

(b) Note that we can rewrite 
 \[
f(x)=
\sum_{j=1}^s (\hat{b}^{(j)T}x+\hat{c}_j)^2+\hat h=
x^T\hat{A}x-2\hat{b}^Tx+\hat{c},\]
 where \[
\di\hat{A}=\sum_{j=1}^s \hat{b}^{(j)}\hat{b}^{(j)T},\ 
\hat{b}=
-\sum_{j=1}^s \hat{c}_j\hat {b}^{(j)},\ \hat {c}= \sum_{j=1}^s \hat {c}^2_j +\hat {h}.
\] 
Compared with (\ref{eq:fofsumofsqurefunctionswithh:ma-zhou-2015}),
\eqref{eq:positivefunctionform:ma-zhou-2015}
 and (\ref{eq:defofAbcdandcinsumofsqurefunction:ma-zhou-2015}), 
we see $\hat{A}=A$.
Set $\hat B=({\hat{b}}^{(1)},{\hat{b}}^{(2)},\cdots, {\hat{b}}^{(s)})$. Then $\hat A=\hat B\hat B^T$ and so
\[ r=\mathrm{rank}(A)= \mathrm{rank}(\hat{A})  = \mathrm{rank}\hat B\le s.\]
This completes the proof. 
 \end{proof}

The result (b) of 
Theorem \ref{cor:sumofsquare:ma-zhou-2015} 
tells the largest number of squares of linearly independent non-constant linear functions in a sum for a non-negative quadratic function. 

When $x=(x_1,\cdots,x_{M-1},t)$, where $t$ denotes time and $x_i, 1\le i\le M-1,$ are  
spatial variables,
positive quadratic function solutions determined by \eqref{eq:concretePQFs:ma-zhou-2015}
with a non-zero $(M,M)$ minor of $A$ 
 lead to lump solutions, and otherwise, lump-type solutions to the corresponding nonlinear equations under either of the two transformations in \eqref{eq:Tranfs:ma-zhou-2015}.

\section{Applications to generalized KP and BKP equations}

\subsection{Generalized KP equations in $(N+1)$-dimensions}

Let us first consider the generalized Kadomtsev-Petviashvili (gKP) equations in $(N+1)$-dimensions:
\be\label{eq:gKP:ma-zhou-2015}
(u_t+6uu_{x_1}+u_{x_1x_1x_1})_{x_1}+\sigma (u_{x_2x_2}+u_{x_3x_3}+\cdots+u_{x_Nx_N})=0,
\ee
where $\sigma =\mp1$ and $N\ge 2$. When $\sigma =-1$, it is called the gKPI equation, and when $\sigma =1$, the gKPII equation.

Denote $x=(x_1,x_2,\cdots,x_N,t)^T\in \mathbb{R}^{N+1}$. Take 
a positive quadratic function:
\be
f(x)=x^TAx+d
\label{eq:qfsolstogKP:ma-zhou-2015}
\ee with $A=A^T\in \mathbb{R}^{(N+1)\times (N+1)}$, $A\ge 0$ and $ d>0 $.
For any $x\in  \mathbb{R}^{N+1}$, the rational function
\[\di u=2(\ln f)_{x_1x_1}=\frac {2(ff_{11}-f_1^2)}{f^2}\] is analytical in $\mathbb{R}^{N+1}$.
Substituting it into (\ref{eq:gKP:ma-zhou-2015}), we have
$$
\begin{array}{l}
 \quad \displaystyle  (u_t+6uu_{x_1}+u_{x_1x_1x_1})_{x_1}+\sigma \sum_{j=2}^Nu_{x_jx_j}\vspace{2mm} \\
\displaystyle
=
\frac{\partial^2}{\partial x_1^2}\Bigl[f^{-2}(D_1^4+D_1D_{N+1}+\sigma  \sum_{j=2}^N D_j^2)f\cdot f\Bigr]=0,
\ \sigma =\mp 1, 
\end{array}
$$
where $D_{N+1}$ is the Hirota bilinear derivative with respect to 
time 
$t$.
Therefore, if $f$ solves the bilinear gKPI or gKPII equation:
\be\label{eq:BgKP:ma-zhou-2015}
(D_1^4+D_1D_{N+1}+\sigma  \sum_{j=2}^N D_j^2)f\cdot f=0,\ \sigma=\mp 1,
\ee
then $\di u=2(\ln f)_{x_1x_1}$ solves the gKPI or gKPII equation in (\ref{eq:gKP:ma-zhou-2015}).
Such a solution process provides us with lump or lump-type solutions to the gKPI or gKPII  equation.

{\theorem
A positive quadratic function $f$ defined by (\ref{eq:qfsolstogKP:ma-zhou-2015}) solves the bilinear gKPI or gKPII equation by (\ref{eq:BgKP:ma-zhou-2015}) if and only if
 \be\label{eq:cond1forgKP:ma-zhou-2015}
 6a_{11}^2+d \tilde{a}=0,
 \ee
and
\be\label{eq:cond2forgKP:ma-zhou-2015}
 \tilde{a}A-(A_1A_{N+1}^T+A_{N+1}A_1^T)-2\sigma \sum_{i=2}^NA_iA_i^T=0,
 \ee
where
 \be \di\tilde{a}:=a_{1N+1}+\sigma \sum_{i=2}^Na_{ii}\le 0.
\label{eq:tilde{a}smallerthanzero:ma-zhou-2015}
\ee
}
\begin{proof}
An application of
 Theorem \ref{thm:condsofpositivesols:ma-zhou-2015}
to the bilinear gKPI and gKPII equations in (\ref{eq:BgKP:ma-zhou-2015})
 tells (\ref{eq:cond1forgKP:ma-zhou-2015}) and (\ref{eq:cond2forgKP:ma-zhou-2015}).
The property $\tilde {a}\le 0$ in 
\eqref{eq:tilde{a}smallerthanzero:ma-zhou-2015}
follows from 
(\ref{eq:cond1forgKP:ma-zhou-2015}) and $ d >0$.
The proof is finished.
 \end{proof}

If $\tilde{a}=0$, then we have
 $a_{11}=0$ by   (\ref{eq:cond1forgKP:ma-zhou-2015}).
Since $A\ge 0$, we have $a_{1,N+1}=0$. Further
$$\sigma \sum_{i=2}^Na_{ii} = \di\tilde{a}- a_{1N+1}
=0.$$
However, $\sigma \ne 0$ and $a_{ii}\ge 0$ for $i=1,\cdots,N+1.$
 Thus, $a_{22}=\cdots=a_{NN}=0$, and there exists only a non-zero solution $A=(a_{ij})_{(N+1)\times (N+1)} $ 
with all $ a_{ij}=0$ except $a_{N+1,N+1}$. The corresponding solution  is $u=2(\ln f)_{x_1x_1} \equiv 0$, a trivial solution.

Now let us introduce
\be 
\ B=2\bar P^{(2)}= \left[
  \begin{array}{ccc}
    0 & 0  & 1 \\
    0 & 2\sigma  I_{N-1} &  0 \\
    1 & 0 &  0 \\
  \end{array}
\right]_{(N+1)\times (N+1)},
\ee 
where $I_{N-1}$ is the identity matrix of size $N-1$, and then the algebraic equation (\ref{eq:cond2forgKP:ma-zhou-2015}) can be written in a compact form:
\be\label{eq:transformedcond2forgKP:ma-zhou-2015}  \tilde{a}A-ABA=0,\ee
where $\tilde a $ is defined by \eqref{eq:tilde{a}smallerthanzero:ma-zhou-2015}.

{\corollary If a positive-semidefinite matrix $A$ satisfies the condition (\ref{eq:transformedcond2forgKP:ma-zhou-2015}), then
 $|A|=0$.
}

\begin{proof}
  If $|A|\ne 0$, then $\tilde{a}I_{N+1}-AB=0$, and so $A=\tilde{a} B^{-1}.$ The matrix $B$ has two eigenvalues $\pm 1$ (and an eigenvalue $2\sigma $ of multiplicity $N-1$), and thus $B^{-1}$ also has two eigenvalues $\pm 1$. Therefore, $A$ is not positive-semidefinite unless $\tilde{a}=0$. In this case, $ABA=0$, and
then $|ABA|=|A|^2|B|= 0$, which leads to $|A|=0$. A contradiction!    \end{proof}

This corollary is also a consequence of Theorem \ref{thm:condsofAnonsingular:ma-zhou-2015}.
For the ($N+1$)-dimensional KP equations, since the corresponding symmetric coefficient matrix $P^{(2)}$, defined by \eqref{eq:defofP^{(2)}:ma-zhou-2015},
 is not zero, their corresponding Hirota bilinear equations in \eqref{eq:BgKP:ma-zhou-2015}
do not possess 
any quadratic function solution which can be written as 
 a sum of squares of 
$N+1$
linearly independent linear functions.

We remark that it is not easy to find all solutions to the 
system of quadratic equations
in (\ref{eq:transformedcond2forgKP:ma-zhou-2015}). 
The following examples show us that the gKPI equations have lump or lump-type solutions.
It is also direct to observe
that any lump or lump-type solution to an ($N+1$)-dimensional gKPI equation is a lump-type solution to an ($(N+1)+1$)-dimensional gKPI equation of the same type as well.

\begin{example} Let us consider the simplest case: $N=2$. This corresponds to the (2+1)-dimensional KPI and KPII equations:
\be\label{eq:gKPin(2+1)d:ma-zhou-2015}
(u_t+6uu_{x}+u_{xxx})_{x}+\sigma u_{yy}=0,\ \sigma=\mp 1,
\ee
where we set $x_1=x$ and $x_2=y$. By using Maple, 
 we can have
\[ A=\left[
              \begin{array}{ccc}
                a & b & \sigma  (ac-2b^2)/a \\
                b & c & -\sigma  bc/a \\
                \sigma  (ac-2b^2)/a & -\sigma bc/a & \sigma ^2 c^2/a \\
              \end{array}
            \right]\  \textrm{with}\  a>0,\ c> 0,\  ac-b^2>0.
\]
This leads to
            \begin{eqnarray}&&
              f(x,y,t)=ax^2+cy^2+\frac{\sigma ^2 c^2}{a}t^2+2bxy-\frac {2\sigma bc }{a}yt+\frac{2\sigma }{a}(ac-2b^2)xt+d \qquad \  \nonumber
\\
              &&\qquad \quad\ \ \,
=a[x+\frac ba y+\frac{\sigma }{a^2}(a c -2b^2)t]^2+\frac{a c -b^2}{a}(y-\frac{2\sigma  b}{a}t)^2+  d ,
            \end{eqnarray}
which reduces to 
            \[
            f(x,y,t)=ax^2+ c y^2+\frac{ \sigma ^2 c ^2}{a}t^2+2\sigma  c xt+ d =a(x+\frac {\sigma  c t}{a})^2+ c y^2+ d ,
             \]
when $b=0$. The condition (\ref{eq:cond1forgKP:ma-zhou-2015}) now reads 
\[
 6a^2+ d [\frac {\sigma  (a c -2b^2)}a+\sigma  c]=6a^2+ 2d\frac { \sigma  (a c -b^2)}a=0,
\]
which yields
\be\label{eq:d_KPI:ma-zhou-2015}
 d =-\frac{3 a^3}{\sigma (a c -b^2)}>0.
\ee

 By Corollary \ref{cor:constant:ma-zhou-2015}, for any constants $\g_1,\g_1,\g_3\in \mathbb{R}$, 
we have the following quadratic function solutions:
 \bea
        &&f(x,y,t)= a[(x-\g_1)+\frac ba (y-\g_2)+\frac{\sigma }{a^2}(a c -2b^2)(t-\g_3)]^2\nonumber\\
              &&\ \ \qquad\qquad+\frac{a c -b^2}{a}[(y-\g_2)-\frac{2\sigma  b}{a}(t-\g_3)]^2+  d\nonumber\\
              &&\qquad \qquad=  a[x+\frac ba y+\frac{\sigma }{a^2}(a c -2b^2)t-\delta_1]^2
            \nonumber\\
              &&\qquad \qquad \ \ 
 +\frac{a c -b^2}{a}(y-\frac{2\sigma  b}{a}t-\delta _2)^2+  d,
\eea
with $\delta _1$ and $\delta _2$ being defined by
$$\delta _1=\g_1+\frac ba\g_2+\frac{\sigma }{a^2}(a c -2b^2)\g_3,\   \delta _2=\g_2-\frac{2\sigma  b}{a}\g_3.$$
Because $\g_1,\g_2,\g_3$ are arbitrary, so are $\delta _1$ and $\delta _2$. Furthermore, the corresponding lump solutions to the (2+1)-dimensional KPI equation
in \eqref{eq:gKPin(2+1)d:ma-zhou-2015}
 read
\[
\begin{array}{l}
\quad u(x,y,t)=2(\ln f)_{xx}
\vspace{2mm}
\\
 =\di\frac {\di 4\big\{  -a^2[x+\frac ba y+\frac{\sigma }{a^2}(a c -2b^2)t-\delta _1]^2
              +{(a c -b^2)}(y-\frac{2\sigma  b}{a}t-\delta _2)^2+  a d\big\}}{\big\{\di  a[x+\frac ba y+\frac{\sigma }{a^2}(a c -2b^2)t-\delta _1]^2
              +\frac{a c -b^2}{a}(y-\frac{2\sigma  b}{a}t-\delta _2)^2+  d\big\}^2},
\end{array}
\]
where $d$ is defined by (\ref{eq:d_KPI:ma-zhou-2015}), $a, b, c\in \mathbb{R}$ satisfy $a>0,\, c>0,\, ac-b^2>0$, and $\delta _1$ and $\delta _2$ are arbitrary. When taking
\[
a=1, \ b=\sqrt{3} \,a, \ c=3(a^2+b^2), \ d=\frac 1{b^2}, \ \delta _1=\delta_2=0,\ y\to \frac 1 {\sqrt{3}}y,
\]
 the resulting lump solutions reduce to the solutions in (\ref{eq:lumpofKPI:ma-lump-jm-2015}).
\end{example}

\begin{remark}
 The condition in (\ref{eq:d_KPI:ma-zhou-2015}) implies that $\sigma =-1$ in order to have lump solutions generated from positive quadratic functions. This shows that the 
(2+1)-dimensional 
KPI equation ($\sigma =-1$) possesses the discussed lump solutions whereas the (2+1)-dimensional  KPII equation
($\sigma =1$) does not.
\end{remark}

\begin{example} We consider the (3+1)-dimensional gKPI equation
\be\label{eq:KP3:ma-zhou-2015}
(u_t+6uu_{x}+u_{xxx})_{x}- u_{yy}-u_{zz}=0.
\ee
By using Maple, we have two classes of lump-type solutions below. Moreover, we will prove that there is no lump solution from quadratic functions.

\flushleft Case I - Sum of two squares: In this case, by Maple, we can have
\be f(x,y,z,t)=(f_1(x,y,z,t))^2+( f_2(x,y,z,t))^2+d,\ee 
with
\be 
 \left\{ \begin{array}{l}
     f_1(x,y,z,t)= x+l_1y+m_1z+\omega_1t -\delta _1,\\[3pt]
     f_2(x,y,z,t)=k_2x+l_2y+m_2z+ \omega_2t -\delta _2,
  \end{array}
 \right.
\label{eq:KP3_fcaseI:ma-zhou-2015}
\ee 
where $k_2,l_1,l_2,m_1,m_2,\delta _1,\delta _2\in\mathbb{R}$ are arbitrary, $l_1m_2 \ne l_2m_1$ and
$$\begin{array}{l}
  \di \omega_1=
\frac{2k_2(l_1l_2+m_1m_2)+(l_1^2-l_2^2)+(m_1^2-m_2^2) }{k_2^2+1},
\\[10pt]
 \di \omega_2=-\frac{k_2[(l_1^2-l_2^2)+(m_1^2-m_2^2)]-2(l_1l_2+m_1m_2)}{k_2^2+1},
\\[10pt]
 \di  d=\frac{3(k_2^2+1)^3}{(k_2l_1-l_2)^2+(k_2m_1-m_2)^2}.\\
 \end{array}$$
The corresponding lump-type solutions read
\[ 
u(x,y,z,t)=\frac{4[(1+k_2^2)d+(k_2^2-1)(f_1^2-f_2^2)-4k_2f_1f_2]}
{(f_1^2+f_2^2+d)^2},
\]
where $f_1$ and $f_2$ are defined by (\ref{eq:KP3_fcaseI:ma-zhou-2015}).\\

\flushleft Case II - Sum of three squares: In this case, by Maple, we can have
\be f(x,y,z,t)=(f_1(x,y,z,t))^2+( f_2(x,y,z,t))^2+( f_3(x,y,z,t))^2+d,\ee
with
\be
 \left\{ \begin{array}{l}
      f_1(x,y,z,t)= x+l_1y+m_1z+ \omega_1t -\delta _1,\\[3pt]
     f_2(x,y,z,t)=k_2x+l_2y+m_2z+ \omega_2t -\delta _2,\\[3pt]
 f_3(x,y,z,t)=l_3y+m_3z + \omega_3t -\delta _3,
  \end{array}
 \right.
\label{eq:KP3_fcaseII:ma-zhou-2015}
\ee
where $k_2,l_1,l_2,l_3\ne 0,m_1,m_3,\delta _1,\delta _2,\delta _3\in\mathbb{R}$ are arbitrary, and
$$\begin{array}{l}
  \di \omega_1=-\frac{\rho _1}{(k_2^2+1)l_3^2},\ 
 \di \omega_2=\frac {\rho _2}{(k_2^2+1)l_3^2} ,\ 
  \di\omega_3=\frac{2\rho_3
}{(k_2^2+1)l_3},\\ [10pt] 
\di m_2=\frac{-k_2l_1m_3+k_2l_3m_1+l_2m_3}{l_3},\ 
  \di
 d=\frac{3(k_2^2+1)^3l_3^2}{(l_3^2+m_3^2)[(k_2l_1-l_2)^2+(k_2^2+1)l_3^2]},
 \end{array}$$
with
$$\begin{array}{l}
\di \rho _1=
k_2^2(l_1^2m_3^2-l_3^2m_1^2)
-l_3^2(l_1^2+m_1^2)+(l_2^2 -2k_2l_1l_2 +l_3^2) (l_3^2+m_3^2) ,
\\  [5pt]
\di\rho _2=
k_2^3(l_1m_3-l_3m_1)^2-2k_2^2 l_2 (l_1m_3-l_3m_1)m_3
\\[5pt]
 \di\qquad
-k_2(l_1^2l_3^2
-l_2^2l_3^2
+2l_1l_3m_1m_3-l_2^2m_3^2-l_3^2m_1^2+l_3^2m_3^2+l_3^4)
\\[5pt]
 \di\qquad 
+2l_2l_3(l_1l_3+m_1m_2)
,\\[5pt]
\rho_3 = 
{-k_2^2(l_1m_3-l_3m_1)m_3
  +k_2l_2(l_3^2+m_3^2)+l_3(l_1l_3+m_1m_3)}.
 \end{array}$$
The corresponding lump-type solutions read
\[ u(x,y,z,t)=\frac{4[(1+k_2^2)(f_3^2+d)+(k_2^2-1)(f_1^2-f_2^2)-4k_2f_1f_2]}
{(f_1^2+f_2^2+f_3^2+d)^2},
\]
where $f_1, f_2$ and $f_3$ are defined by (\ref{eq:KP3_fcaseII:ma-zhou-2015}).
\end{example}

The formula for $m_2$ in the above example means the corresponding first minor $M_{44}$ 
is zero, and so, the presented solution is not a lump solution. 
Generally, when $N\ge 3$, there is no solution to the matrix equation \eqref{eq:transformedcond2forgKP:ma-zhou-2015} 
with a non-zero first minor 
 $M_{N+1,N+1}$, indeed. Therefore, the above gKP equations in $(N+1)$-dimensions with $N\ge 3$
have 
no lump solutions generated from quadratic functions. We prove a more general result as follows.

\begin{theorem}Let $N\ge 3$.
Then there is no symmetric matrix solution $A\in \mathbb{R}^{(N+1)\times (N+1)}$ to the matrix equation \eqref{eq:transformedcond2forgKP:ma-zhou-2015} with rank$(A)=N$, which implies that 
the $(N+1)$-dimensional gKP equations (\ref{eq:gKP:ma-zhou-2015})
have no lump solution generated from quadratic functions under the transformation $u=2(\ln f)_{xx}$.
\end{theorem}

\begin{proof}
Suppose that there is a symmetric matrix $A\in \mathbb{R}^{(N+1)\times (N+1)}$ which solves the equation \eqref{eq:transformedcond2forgKP:ma-zhou-2015} and whose rank is $N$.
Then, since $A$ is symmetric and rank$(A)=N$, there exists an orthogonal matrix 
$U\in \mathbb{R}^{(N+1)\times (N+1)}$ such that 
\[ \hat A = U^T A U = \left [\begin{array}{cc}
\hat A_1 & 0 \vspace{2mm}\\
0& 0
\end{array}\right],\ \hat A_1 =\diag (\lambda _1,\cdots ,\lambda_N), 
\]
where $\lambda _i\ne 0,\ 1\le i\le N$.
Set 
\[ \hat B = U^T B U = \left [\begin{array}{cc}
\hat B_1 & \hat B_2 \vspace{2mm}\\
\hat B_3& \hat B_4
\end{array}\right],\ \hat B_1 =(\hat b_{ij})_{N\times N}  \in \mathbb{R}^{N\times N}.
\]
Upon noting that $\tilde a$ is an invariant under an orthogonal similarity transformation, it follows from \eqref{eq:transformedcond2forgKP:ma-zhou-2015} 
that
\[ \tilde a \hat A_1-\hat A_1\hat B_1\hat A_1 =0 ,\ \tilde a =\sum_{i,j=1}^{N+1}a_{ij}p_{ij}=\frac 12  \sum_{k=1}^N \lambda _k \hat b_{kk}.
 \]
 Then,  based on this sub-matrix equation, using the same idea in the proof of Theorem \ref{thm:condsofAnonsingular:ma-zhou-2015} shows that 
$ N \tilde a = 2 \tilde a $, which leads to $\tilde a =0$ since $N\ge 3$. Further, we have $\hat B_1=0$, and thus, rank$(\hat B)\le 2$, which is a contradiction to rank$(\hat B)=\textrm{rank}(B)=N+1$. 
Therefore, there is no symmetric matrix solution $A$ to the equation \eqref{eq:transformedcond2forgKP:ma-zhou-2015} with rank$(A)=N$. 

Finally, note that
the existence of a non-zero $(N+1,N+1)$ minor 
 of $A$ implies that rank$(A)\ge N$, and thus, by Theorem \ref{thm:condsofAnonsingular:ma-zhou-2015}, we have rank$(A)=N$. 
Now, it follows that there is no symmetric matrix solution $A$ to the equation \eqref{eq:transformedcond2forgKP:ma-zhou-2015} 
with a non-zero $(N+1,N+1)$ minor.
This means that the gKP equations, defined by (\ref{eq:gKP:ma-zhou-2015}), in $(N+1)$-dimensions with $N\ge 3$
have 
no lump solution, which are generated from quadratic functions under the transformation $u=2(\ln f)_{xx}$.  
The proof is finished.
\end{proof}

\subsection{Generalized KP and BKP equations with general 2nd-order derivatives}

Let us next consider generalized KP and BKP equations with a general sum of second-order Hirota derivative terms. We will present lump solutions to those two generalized KP and BKP equations.

\begin{example}
We consider the 
following generalized KP (gKP) equation:
\be\label{eq:ggKP:ma-zhou-2015}
K_{gKP1}(u)=(6uu_{x}+u_{xxx})_{x}+
c_1u_{xx}+2c_2u_{xy}+2c_3u_{xt}
+c_4u_{yy}+2c_5u_{yt}+c_6u_{tt}
=0,
\ee
with arbitrary constant coefficients $c_i,\ 1\le i\le 6.$
Under the typical transformation $u=2(\ln f)_{xx}$,
this general equation itself has a Hirota bilinear form: 
\be B_{gKP1}(f)=
( D_x^4+ c_1
D_x^2+2c_2D_xD_y+2c_3D_xD_t+c_4D_y^2+2c_5D_yD_t+c_6D_t^2
 )f\cdot f=0,\label{eq:gbKPI:ma-zhou-2015}
\ee
since we have
\[ 
K_{gKP1}(u)=\Bigl(\frac {B_{gKP1}(f)}{f^2}\Bigr)_{xx}.
\]
The equation \eqref{eq:ggKP:ma-zhou-2015} reduces to the (2+1)-dimensional KPI and KPII equations in \eqref{eq:gKPin(2+1)d:ma-zhou-2015},
upon taking \[c_1=c_2=0, \ c_3 =\frac 12 ,\ c_4=\sigma ,\ c_5=c_6=0.\]

To search for quadratic function solutions
to the (2+1)-dimensional 
bilinear gKP equation \eqref{eq:gbKPI:ma-zhou-2015},
we start with
\be
f=g_1^2+g_2^2+a_{9},
\
g_1=a_1 x +a_2 y +a_3 t+a_4,\
g_2=a_5 x +a_6 y  +a_7 t+a_{8},
\label{eq:defoffforggKP:ma-zhou-2015}
\ee
where $a_i$, $1\le i\le 9$, are real parameters to be determined.
A direct Maple symbolic computation with this function $f$ 
generates the set of three constraining equations for the parameters and the coefficients:
\be
\left\{\begin{array}{l}
\displaystyle   a_{{9}}=  
-\frac {3(a_1^2+a_5^2)^3}
{
\mu _1 c_{{4}}+\mu _2 c_{{5}}+ \mu _3 c_{{6}}
},
\vspace{2mm}
\\
\displaystyle
c_1=\frac {\nu _{1,1}c_3+\nu _{1,2}c_4+\nu _{1,3}c_5+\nu _{1,4}c_6 }{(a_1^2+a_5^2) (a_{{1}}a_{{6}}-a_{{2}}a_{{5}})},
\vspace{2mm}\\
\displaystyle
c_2=-\frac {\nu _{2,1}c_3+\nu _{2,2}c_4+\nu _{2,3}c_5+\nu _{2,4}c_6 }{(a_1^2+a_5^2) (a_{{1}}a_{{6}}-a_{{2}}a_{{5}})}
,
\end{array}\right. \ee 
with
\be 
\left\{ \begin{array}
{l}
\displaystyle
\mu _1=( a_{{1}}a_{{6}}-a_{{2}}a_{{5}} ) ^{2}
,\ 
\vspace{2mm}
\\
\displaystyle
\mu _2=
2 ( a
_{{1}}a_{{6}}-a_{{2}}a_{{5}} )  ( a_{{1}}a_{{7}}-a_{{3}}a_{
{5}} )
,\ 
\vspace{2mm}
\\
\displaystyle \mu _3=( a_{{1}}a_{{7}}-a_{{3}}a_{{5}} ) ^{2
}
,
\end{array}
\right.
\label{eq:defformu_i:ma-zhou-2015}
\ee 
\be
\left\{ \begin{array}
{l}
\displaystyle
\displaystyle
\nu _{1,1}=
2 ( {a^2_{{1}}}+{a^2_{{5}}} )  ( a_{{2}}a_{{7}}-
a_{{3}}a_{{6}} )
,\ 
\vspace{2mm}
\\
\displaystyle
\nu _{1,2}=
( {a^2_{{2}}}+{a^2_{{6}}})  ( a_{{1}}a_{{6}}-a_{
{2}}a_{{5}} 
)
,\ 
\vspace{2mm}
\\
\displaystyle
\nu _{1,3}=
2 ( {a^2_{{2}}}+{a^2_{{6}}} )  ( a_{{1}}a_{{7}}-
a_{{3}}a_{{5}}
),\ 
\vspace{2mm}
\\
\displaystyle
\nu _{1,4}=
2 a_{{3}}a_{{7}} ( a_{{1}}a_{{2}}-a_{{5}}a_{{6}} ) -
  ( {a^2_{{3}}}-{a^2_{{7}}}) 
( a_{{1}}a_{{6}}+a_{{2}}a_{{5}} ) 
\end{array}
\right.
\label{eq:deffornu_1i:ma-zhou-2015}
\ee
and 
\be 
\left\{
\begin{array}{l}
\displaystyle
\nu _{2,1}=( {a^2_{{1}}}+{a^2_{{5}}} )  ( a_{{1}}a_{{7}}-a_{
{3}}a_{{5}} ) 
,\ 
\vspace{2mm}
\\
\displaystyle
\nu _{2,2}= ( a_{{1}}a_{{2}}+a_{{5}}a_{{6}} )  ( a_{{1}}a_{{6}}-a
_{{2}}a_{{5}}
),\ 
\vspace{2mm}
\\
\displaystyle
\nu _{2,3}=( {a^2_{{1}}}-{a^2_{{5}}} )  ( a_{{2}}a_{{7}}+a_{
{3}}a_{{6}} ) -2\,a_{{1}}a_{{5}} ( a_{{2}}a_{{3}}-a_{{6}}a_
{{7}} 
),\ 
\vspace{2mm}
\\
\displaystyle
\nu _{2,4}=
( a_{{1}}a_{{3}}+a_{{5}}a_{{7}} )  ( a_{{1}}a_{{7}}-a
_{{3}}a_{{5}}),
\end{array}\right.
\label{eq:deffornu_2i:ma-zhou-2015}
\ee 
where all involved other parameters and coefficients are arbitrary provided that the expressions make sense.

When a determinant condition 
\be a_{{1}}a_{{6}}-a_{{2}}a_{{5}}=\left|\begin{array} {cc}
a_1 & a_2\vspace{2mm}\\
a_5 & a_6
\end{array} \right| \ne 0,
\label{eq:C1forLumptoggKP:ma-zhou-2015}
 \ee
is satisfied, 
the above quadratic function $f$ will be positive if and only if $a_9>0$, i.e.,
\be  
\mu _1 c_4+\mu _2c_5 +\mu _3 c_6 <0,
\label{eq:C2forLumptoggKP:ma-zhou-2015}
\ee
and
the resulting class of positive quadratic function solutions generates lump solutions to
the (2+1)-dimensional gKPI equation \eqref{eq:ggKP:ma-zhou-2015} through the transformation
$u=2(\ln f)_{xx}$:
\be \label{eq:LumpSolstoggKP:ma-zhou-2015}
u=\frac {4(a_1^2+a_5^2)f-8(a_1g_1+a_5g_2)^2}{f^2},
\ee 
where the functions $f,g_1,g_2$ are determined above.

In this class of lump solutions, all involved eight parameters $a_i,\ 1\le i
\le 8$, and four coefficients $c_i, \ 3\le i\le 6$, are arbitrary provided that the two 
conditions, \eqref{eq:C1forLumptoggKP:ma-zhou-2015} and \eqref{eq:C2forLumptoggKP:ma-zhou-2015},
 are satisfied.
 The determinant condition \eqref{eq:C1forLumptoggKP:ma-zhou-2015}
precisely means that two directions $(a_1,a_2)$ and $(a_5,a_6)$ in the $(x,y)$-plane are not parallel, which guarantees, together with 
  \eqref{eq:C2forLumptoggKP:ma-zhou-2015}, that the resulting solutions in \eqref{eq:LumpSolstoggKP:ma-zhou-2015} are lump solutions.

For the standard KPI and KPII equations in
\eqref{eq:gKPin(2+1)d:ma-zhou-2015},
we have 
\[
a_9= -
 {\frac {3 \sigma ( {a^2_{{1}}}+{a^2_{{5}}} ) ^{3}}{ ( a
_{{1}}a_{{6}}-a_{{2}}a_{{5}} ) ^{2}}},
\] 
and obtain
\[ 
 a_3=- {\frac {\sigma (a_{{1}}{a_{{2}}}^{2}-a_{{1}}{a_{{6}}}^{2}+2\,a_{{2}}a_{{5}}a_{
{6}}) }{{a_{{1}}}^{2}+{a_{{5}}}^{2}}}
,\ 
a_7= - {\frac {\sigma (2\,a_{{1}}a_{{2}}a_{{6}}-{a_{{2}}}^{2}a_{{5}}+a_{{5}}{a_{{6}}}
^{2})}{{a_{{1}}}^{2}+{a_{{5}}}^{2}}},
\]
upon solving the system
\[ c_1= \frac 12 v_{1,1} +\sigma  v_{1,2}=0,\ c_2=\frac 12 v_{2,1} +\sigma  v_{2,2}=0, \]
for $a_3$ and $a_7$.
This exactly produces to the lump solution presented in 
(Ma 2015) for the KPI equation, but the resulting solution to the KPII equation has pole 
singularity in the $(x,y)$-plane at any time, due to $a_9<0$.

\end{example}

\begin{example}
We consider the following generalized BKP (gBKP) equation:
\begin{eqnarray}
&
K_{gBKP}(u)=( 15 u_x^3+15 u_xu_{3x}+u_{5x})_{x}+
c_1  [ u_{3x,y} +3(u_x u_{y})_x  ] 
\nonumber \\
& 
  + 
c_2 u_{xx} 
+2c_3u_{xy}
+2c_4u_{xt}
+c_5u_{yy}+2c_6u_{yt}+c_7u_{tt}
=0,\label{eq:ggBKP:ma-zhou-2015}
\end{eqnarray}
with arbitrary constant coefficients $c_i,\ 1\le i\le 7.$
Under the other typical transformation $u=2(\ln f)_{x}$,
this general equation itself has a Hirota bilinear form: 
\begin{eqnarray}
& B_{gBKP}(f)=
( D_x^6+ c_1
D_x^3D_y +c_2  D_x^2  + 2c_3  D_xD_y
\nonumber
\\
& \ +2c_4D_xD_t+c_5D_y^2+2c_6D_yD_t+c_7D_t^2
 )f\cdot f=0,
\label{eq:gbBKP:ma-zhou-2015}
\end{eqnarray}
since we have
\[ 
K_{gBKP}(u)=\Bigl(\frac {B_{gBKP}(f)}{f^2}\Bigr)_{x}.
\]
The equation \eqref{eq:ggBKP:ma-zhou-2015} reduces to the 
$(2+1)$-dimensional BKP equations: 
\be 
\label{eq:sBKPs:ma-zhou-2015}
(u_t+ 15 u_x^3+15 u_xu_{3x} - 15 u_xu_y +u_{5x})_{x}-5 u_{3x,y}  +5\sigma  u_{yy}=0,\ \sigma=\mp 1, 
\ee
upon taking \be
c_1=-5,\  c_4=\frac 12 ,\ c_5=  5\sigma ,\ c_2=c_3=c_6=c_7=0.
\label{eq:defofC_iforBKP:ma-zhou-2015}
\ee

To search for quadratic function solutions
to the (2+1)-dimensional
bilinear gBKP equation \eqref{eq:gbBKP:ma-zhou-2015},
we begin with the same class of quadratic functions defined by 
\eqref{eq:defoffforggKP:ma-zhou-2015}.
A similar direct Maple symbolic computation with $f$ 
leads to the set of three constraining equations for the parameters and the coefficients:
\be
\left\{\begin{array}{l}
\displaystyle   a_{{9}}=  
-\frac { 3(a_1a_2+a_5a_6)(a_1^2+a_5^2)^2c_1}
{
\mu _1 c_{{5}}+\mu _2 c_{{6}}+ \mu _3 c_{{7}}
},
\vspace{2mm}
\\
\displaystyle
c_2=\frac {\nu _{1,1}c_4+\nu _{1,2}c_5+\nu _{1,3}c_6+\nu _{1,4}c_7 }{(a_1^2+a_5^2) (a_{{1}}a_{{6}}-a_{{2}}a_{{5}})},
\vspace{2mm}\\
\displaystyle
c_3=-\frac {\nu _{2,1}c_4+\nu _{2,2}c_5+\nu _{2,3}c_6+\nu _{2,4}c_7 }{(a_1^2+a_5^2) (a_{{1}}a_{{6}}-a_{{2}}a_{{5}})}
,
\end{array}\right. \ee 
with 
$\mu_i,$ $ 1\le i\le 3,$ $\nu_{1,i},$ $1\le i\le 4$, and $\nu_{2,i},$ $ 1\le i\le 4$, being defined  
by \eqref{eq:defformu_i:ma-zhou-2015}, \eqref{eq:deffornu_1i:ma-zhou-2015} and \eqref{eq:deffornu_2i:ma-zhou-2015}.
All involved parameters and coefficients are arbitrary provided that the expressions make sense.

When we require a determinant condition in
\eqref{eq:C1forLumptoggKP:ma-zhou-2015},
the presented quadratic function $f$ will be positive if and only if $a_9>0$, which means 
\be  
\frac { (a_1a_2+a_5a_6)c_1}
{
\mu _1 c_{{5}}+\mu _2 c_{{6}}+ \mu _3 c_{{7}}}
<0,
\label{eq:C2forLumptogBKP:ma-zhou-2015}
\ee
and
the resulting class of positive quadratic function solutions yields lump solutions to
the (2+1)-dimensional gBKP equation \eqref{eq:ggBKP:ma-zhou-2015} through the transformation
$u=2(\ln f)_{x}$:
\be \label{eq:LumpSolstogBKP:ma-zhou-2015}
u=\frac {4(a_1g_1+a_5g_2)}{f},
\ee 
where the functions $f,g_1,g_2$ are defined above.

In this presented class of lump solutions, all involved eight parameters $a_i,\ 1\le i
\le 8$, and five coefficients $c_1,\,c_i,\ 4\le i\le 7$, are arbitrary provided that the two 
conditions, \eqref{eq:C1forLumptoggKP:ma-zhou-2015} and \eqref{eq:C2forLumptogBKP:ma-zhou-2015},
 are satisfied.
 The determinant condition \eqref{eq:C1forLumptoggKP:ma-zhou-2015}
exactly requires that two directions $(a_1,a_2)$ and $(a_5,a_6)$ in the $(x,y)$-plane are not parallel, which similarly guarantees, together with 
  \eqref{eq:C2forLumptogBKP:ma-zhou-2015}, that the presented solutions in \eqref{eq:LumpSolstogBKP:ma-zhou-2015} are lump solutions.

The coefficient constraints \eqref{eq:defofC_iforBKP:ma-zhou-2015}
engender the standard BKP equations in 
\eqref{eq:sBKPs:ma-zhou-2015}.
In this case, similarly we have 
\[
a_9= 
  {\frac {3\sigma (a_1a_2+a_5a_6) ( {a^2_{{1}}}+{a^2_{{5}}} ) ^{2}}{ ( a
_{{1}}a_{{6}}-a_{{2}}a_{{5}} ) ^{2}}},
\] 
and obtain
\[ 
a_3= - {\frac {5\sigma (a_{{1}}{a_{{2}}}^{2}-a_{{1}}{a_{{6}}}^{2}+2\,a_{{2}}a_{{5}}a_{{6}})}
{{a_{{1}}}^{2}+{a_{{5}}}^{2}}}
,\ 
a_7= - {\frac {5\sigma (2\,a_{{1}}a_{{2}}a_{{6}}-{a_{{2}}}^{2}a_{{5}}+a_{{5}}{a_{{6}}}
^{2})}
{{a_{{1}}}^{2}+{a_{{5}}}^{2}}},
\]
upon solving the system
\[ c_2= \frac 12 v_{1,1} + 5 \sigma  v_{1,2}=0,\ c_3=\frac 12 v_{2,1} + 5 \sigma v_{2,2}=0, \]
for $a_3$ and $a_7$.
This generates lump solutions to the BKP equations in 
\eqref{eq:sBKPs:ma-zhou-2015}, when $(a_1a_2+a_5a_6) <0$ with 
minus
sign ($\sigma =-1$)
or $(a_1a_2+a_5a_6)>0$ with plus
sign ($\sigma =1$). This solution phenomenon is pretty different from 
what we presented for the standard KP equations in the previous example. 
The solution case with minus sign also covers the lump solution presented in (Gilson and Nimmo 1990).
Actually, if we take 
\[
a_1=1, \ a_2=3(\alpha ^2-\beta ^2), \ a_5=0,\ a_6=6\alpha \beta, 
\]
where $\alpha $ and $\beta $ are arbitray but $\beta ^2<\alpha ^2$,
which leads to 
\[ a_3= 45(\alpha ^4-  6 \alpha ^2\beta ^2 +\beta ^2),\ a_7=180\alpha \beta (\alpha ^2-\beta ^2),
 \ a_9= \frac {\beta ^2 -\alpha ^2}{4\alpha ^2 \beta ^2}, \] 
then the resulting lump solution is exactly the one in (Gilson and Nimmo 1990). 

\end{example}


\section{Concluding remarks}

In this paper, we studied positive quadratic function solutions to Hirota bilinear equations. Sufficient and necessary conditions 
for the existence of such polynomial solutions were given. In turn, positive quadratic function solutions generate lump or lump-type 
solutions to nonlinear partial differential equations possessing Hirota bilinear forms. Applications were made for 
a few generalized KP and BKP equations.

We remark that putting Theorem \ref{thm:positiveness:ma-zhou-2015} and 
Theorem \ref{thm:PQFasSumofSquaresofLFs:ma-zhou-2015} together proves 
Hilbert's 17th problem for quadratic functions, 
but the conjecture is not true for higher-order polynomial functions.
Moreover, 
Theorem \ref{thm:positiveness:ma-zhou-2015} provides a criterion for the positivity of quadratic functions.
It, however, still remains open how to determine the positivity 
of higher-order multivariate polynomials,
which is a further problem of Hilbert's 17th problem. 
It should be also interesting to look for positive polynomial solutions to generalized bilinear equations (Ma 2011), 
which generate exact rational function solutions to novel types of nonlinear differential equations
(Shi et al. 2015; Zhang Y and Ma 2015; Zhang YF and Ma 2015). The first example of such solutions one can try
could be positive quartic function solutions.

It is evident that
if $A$ is positive-definite in the quadratic function $f$ defined by (\ref{eq:concretePQFs:ma-zhou-2015}), then
$f\to \infty $ when $|x|\to \infty$ in any direction in $\mathbb{R}^M$.
This guarantees that
$u=2(\ln f)_{x_1}\to 0$ and $u=2(\ln f)_{x_1x_1}\to 0$ 
as $|x|\to \infty$
in any direction in $\mathbb{R}^M$, and so they yield lump solutions, rationally localized solutions in all directions in space and time.
If $A$ is positive-semidefinite, then $u=2(\ln f)_{x_1}$ and $u=2(\ln f)_{x_1x_1}$
do not go to zero in all directions in 
$\mathbb{R}^M$
 but may go to zero
in all directions in a subspace of $\mathbb{R}^M$.
Therefore, they usually lead to lump-type solutions, and lump solutions if the subspace is the actual space 
which the spatial variables belong to. Three of such examples about the generalized KP and BKP equations were just discussed.

Through Theorem \ref{thm:positiveness:ma-zhou-2015}, we can obtain a by-product, which partially answers an open question in (Ma 2013; Ma et al. 2012):
how to determine a real multivariate polynomial which has only one zero?
By Theorem \ref{thm:positiveness:ma-zhou-2015}, a quadratic function $f$ has only one zero at $\alpha \in \mathbb{R}^M$ if and only if
\[f(x)= (x-\alpha )^TA(x-\alpha ),\ x\in \mathbb{R}^M,\]
where $A\in \mathbb{R}^{M\times M}$ is positive-definite or negative-definite, since a 
multivariate polynomial is either non-negative or non-positive if it has one zero.
If $A$ above is only positive-semidefinite, then there must exist infinitely many zeros.

Moreover, Theorem \ref{thm:positiveness:ma-zhou-2015} tells that
if a quadratic function $f$ is positive on $\mathbb{R}^M$, i.e., 
$f(x)>0$ for all $x\in \mathbb{R}^M$, then
there is a positive constant $d$ such that $f(x)\ge d$ for all $x\in \mathbb{R}^M$.
But this is not the case for higher order multivariate polynomials. There are counterexamples:
$$f_{mn}(x,y)=x^{2m}+(x^ny^n-1)^2,\quad m,n\in \mathbb{N},$$
for which $f_{mn}(x,y)>0$  since $f_{mn}(0,y)=1$ and $f_{mn}(x,y)\ge x^{2m}>0$ for $x\ne 0$. It is apparent that 
 $\di \lim_{k\to \infty}f_{mn}(k^{-1},k)=0$. This clearly implies that $f_{mn}$ cannot be bounded from below by any positive constant.

\vskip 3mm

\small
\noindent{\it Acknowledgments:}
The work was supported in
part by 
 NSFC under the grants 11371326 and 11271008.
The authors would also like to thank X. Gu,  X. Lu, H. C. Ma, S. Manukure, M. Mcanally, B. Shekhtman,  F. D. Wang, X. L. Yang, Y. Q. Yao and Y. J. Zhang for their valuable discussions about lump solutions.




\end{document}